\documentclass[a4paper,UKenglish,cleveref, thm-restate]{lipics-v2021}
\nolinenumbers
\hideLIPIcs  %

\title{Univalent Monoidal Categories} %

\author{Kobe Wullaert}{Delft University of Technology, The Netherlands \and\url{https://kfwullaert.github.io/}}{K.F.Wullaert@tudelft.nl}{https://orcid.org/0000-0003-4281-2739}{}

\author{Ralph Matthes}{IRIT, Université de Toulouse, CNRS, Toulouse INP, UT3, Toulouse, France \and\url{https://www.irit.fr/~Ralph.Matthes/}}{Ralph.Matthes@irit.fr}{https://orcid.org/0000-0002-7299-2411}{}

\author{Benedikt Ahrens}{Delft University of Technology, The Netherlands \and University of Birmingham, United Kingdom \and \url{https://benediktahrens.gitlab.io} }{B.P.Ahrens@tudelft.nl}{https://orcid.org/0000-0002-6786-4538}{This work was
partially funded by EPSRC under agreement number EP/T000252/1.}%

\authorrunning{K. Wullaert, R. Matthes, and B. Ahrens} %

\Copyright{Kobe Wullaert, Ralph Matthes, and Benedikt Ahrens} %



\acknowledgements{We gratefully acknowledge the work by the Coq development team in providing the Coq proof assistant and surrounding infrastructure, as well as their support in keeping UniMath compatible with Coq.
  Furthermore, we thank Niels van der Weide for helpful discussions on the subject matter and for reviewing the formalization.
  We also thank Vikraman Choudhury for a question regarding the connection of the monoidal Rezk completion to the Day convolution product.
  Furthermore, we thank the anonymous referees for their helpful feedback; besides triggering further thoughts on the technical aspects, it helped us to improve the presentation.
}%

\usepackage{mathtools}

\usepackage{xifthen}

\newcommand{\Coq}{\href{https://coq.inria.fr}{\nolinkurl{Coq}}\xspace}
\newcommand{\UniMath}{\href{https://github.com/UniMath/UniMath}{\nolinkurl{UniMath}}\xspace}

\newcommand{\longhash}{6d2d288264d28f2d1966bb518de180d73e1c5e47}
\newcommand{\shorthash}{6d2d288}

\newcommand{\coqdocbasebaseurl}{https://benediktahrens.gitlab.io/unimathdoc/\shorthash}

\newcommand{\coqdocbaseurl}{\coqdocbasebaseurl /UniMath.}
\newcommand{\urlhash}{\#}
\newcommand{\coqdocurl}[2]{\coqdocbaseurl #1.html\urlhash #2}
\newcommand{\nolinkcoqident}[1]{\nolinkurl{#1}} %
\makeatletter
\newcommand{\coqident}{\begingroup\@makeother\#\@coqident}
\newcommand{\@coqident}[3][]{%
  \ifthenelse{\isempty{#2}}%
  {\nolinkcoqident{#3}}%
  {\ifthenelse{\isempty{#1}}%
  {\href{\coqdocurl{#2}{#3}}{\nolinkcoqident{#3}}}%
  {\href{\coqdocurl{#2}{#3}}{\nolinkcoqident{#1}}}}%
\endgroup}
\newcommand{\coqfile}[2]{%
  \ifthenelse{\isempty{#1}}%
  {\href{\coqdocbaseurl #2.html}{\nolinkcoqident{#2.v}}}%
  {\href{\coqdocbaseurl #1.#2.html}{\nolinkcoqident{#2.v}}}}
\makeatother

\usepackage[draft]{optional}
\newcommand{\plan}[1]{}
\newcommand{\BA}[1]{}
\newcommand{\RM}[1]{}
\newcommand{\KW}[1]{}
\opt{draft}{
\renewcommand{\plan}[1]{\textcolor{blue}{#1}\PackageWarning{TODO}{TODO: #1}}
\renewcommand{\BA}[1]{\textcolor{orange}{BA: #1}\PackageWarning{TODO}{TODO: #1}}
\renewcommand{\RM}[1]{\textcolor{purple}{RM: #1}\PackageWarning{TODO}{TODO: #1}}
\renewcommand{\KW}[1]{\textcolor{magenta}{KW: #1}\PackageWarning{TODO}{TODO: #1}}
}

\usepackage{tikz-cd}
\usepackage{xspace}

\newcommand{\cfont}[1]{\ensuremath{\mathsf{#1}}}
\newcommand{\RC}{\ensuremath{\cfont{RC}}}

\newcommand{\catfont}[1]{\ensuremath{\mathcal{#1}}}
\newcommand{\C}{\catfont{C}}

\newcommand{\CAT}{\mathbf{Cat}}
\newcommand{\CATU}{\CAT_{\mathit{univ}}}

\newcommand{\CATT}{\CAT_T}
\newcommand{\CATUnit}{\CAT_U}
\newcommand{\CATTU}{\CAT_{TU}}
\newcommand{\CATLU}{\CAT_{LU}}
\newcommand{\CATRU}{\CAT_{RU}}
\newcommand{\CATA}{\CAT_{A}}
\newcommand{\CATUA}{\CAT_{\mathit{UA}}}
\newcommand{\CATP}{\CAT_{P}}
\newcommand{\CATS}{\CAT_{S}}
\newcommand{\CATl}{\CAT_{\ell}}

\newcommand{\restrict}[1]{{#1}\vert_{\mathit{univ}}}
\newcommand{\univCATT}{\restrict{\CATT}}
\newcommand{\univCATUnit}{\restrict{\CATUnit}}

\newcommand{\univCATLU}{\restrict{\CATLU}}
\newcommand{\univCATRU}{\restrict{\CATRU}}
\newcommand{\univCATA}{\restrict{\CATA}}
\newcommand{\univCATUA}{\restrict{\CATUA}}
\newcommand{\univCATP}{\restrict{\CATP}}
\newcommand{\univCATS}{\restrict{\CATS}}
\newcommand{\univCATl}{\restrict{\CATl}}

\newcommand{\MONCAT}{\mathbf{MonCat}}
\newcommand{\MONCATU}{\MONCAT_{\mathit{univ}}}

\newcommand{\MONCATS}{\MONCAT^{\mathit{stg}}}
\newcommand{\MONCATSU}{\MONCATS_{\mathit{univ}}}

\newcommand{\BB}{\catfont{B}}
\newcommand{\CC}{\catfont{C}}
\newcommand{\DD}{\catfont{D}}
\newcommand{\EE}{\catfont{E}}
\newcommand{\HH}{\mathcal{H}}

\newcommand{\ob}[1]{#1_0}
\newcommand{\comp}{\cdot}
\newcommand{\id}[1]{\mathsf{Id}_{#1}}
\newcommand{\mor}[3]{{#1}(#2,#3)}
\newcommand{\lwhisker}{\triangleleft}
\newcommand{\rwhisker}{\triangleright}
\newcommand{\precomp}[1]{#1 \comp (-)}

\newcommand{\tensor}{\otimes}
\newcommand{\unit}{I}
\newcommand{\lu}{\lambda}
\newcommand{\ru}{\rho}
\newcommand{\ass}{\alpha}

\newcommand{\liftstruct}[1]{\hat{#1}}
\newcommand{\tensorD}{\liftstruct{\tensor}}
\newcommand{\unitD}{\liftstruct{\unit}}
\newcommand{\luD}{\liftstruct{\lu}}
\newcommand{\ruD}{\liftstruct{\ru}}
\newcommand{\assD}{\liftstruct{\alpha}}

\newcommand{\pt}[1]{\mathsf{\mu}^{#1}}
\newcommand{\pu}[1]{\mathsf{\epsilon}^{#1}}
\newcommand{\plu}[1]{\mathsf{plu}^{#1}}

\newcommand{\idtoiso}{\ensuremath{\cfont{idtoiso}}}

\newcommand{\idtoisoob}[2]{\ensuremath{\cfont{idtoiso^{2,0}_{#1,#2}}}}
\newcommand{\idtoeq}{\ensuremath{\cfont{idtoeq}}}
\newcommand{\eqtoiso}{\ensuremath{\cfont{eqtoiso}}}
\newcommand{\tensorIso}{\cfont{tensorIso}}
\newcommand{\tensorEq}{\cfont{tensorEq}}

\newcommand{\isiso}[1]{\cfont{isIso(}#1\cfont{)}}

\newcommand{\x}{\bar{x}}
\newcommand{\y}{\bar{y}}
\newcommand{\z}{\bar{z}}
\newcommand{\f}{\bar{f}}
\newcommand{\g}{\bar{g}}

\newcommand{\eg}{e.\,g.\xspace}

\newcommand{\ie}{i.\,e.\xspace}
\newcommand{\wrt}{w.\,r.\,t.\xspace}
\newcommand{\resp}{resp.\xspace}

\newtheorem{problem}[theorem]{Problem}
\newtheorem{notation}[theorem]{Notation}

\DeclareFontFamily{U}{min}{}
\DeclareFontShape{U}{min}{m}{n}{<-> udmj30}{}
\newcommand\yon{\!\text{\usefont{U}{min}{m}{n}\symbol{'210}}\!}

\begin{document}

\maketitle

\begin{abstract}
  \emph{Univalent} categories constitute a well-behaved and useful notion of category in univalent foundations.
  The notion of univalence has subsequently been generalized to bicategories and other structures in (higher) category theory.
  Here, we zoom in on monoidal categories and study them in a univalent setting.
  Specifically, we show that the bicategory of univalent monoidal categories is univalent.
  Furthermore, we construct a Rezk completion for monoidal categories: we show how any monoidal category is weakly equivalent to a univalent monoidal category, universally.
  We have fully formalized these results in UniMath, a library of univalent mathematics in the Coq proof assistant.
\end{abstract}

\section{Introduction}
\label{sec:introduction}

When working in univalent foundations (see~\cite{hottbook}), definitions have to be designed carefully in order to correspond, via the intended semantics, to the \emph{expected} notions in set-theoretic foundations.
The notion of univalent category \cite{DBLP:journals/mscs/AhrensKS15} has been shown to be a good notion, in the sense that it corresponds to the usual notion of category under Voevodsky's model in simplicial sets \cite{kapulkinSimplicialModelUnivalent2021}.%
\footnote{To emphasize that univalent categories are the right notion of category in univalent foundations, they are just called ``categories'' in \cite{DBLP:journals/mscs/AhrensKS15}.}
Examples of univalent categories are plentiful, but not all categories arising in practice---for instance when studying categorical semantics of type theory---are univalent. %
In  \cite{DBLP:journals/mscs/AhrensKS15}, the authors give a construction of a ``free'' univalent category from any category $\C$, which they call the Rezk completion of $\C$.

Since then, the univalence condition and completion operation have been studied further.

Firstly, in \cite{phd/NVDW}, Van der Weide constructs a class of higher inductive types using the groupoid quotient.
It is shown that the groupoid quotient gives rise to a biadjunction between the bicategory of groupoids and the bicategory of $1$-types (which is isomorphic to the bicategory of univalent groupoids);
the left adjoint thus yields a univalent completion operation for groupoids.
Van der Weide furthermore lifts this completion to ``structured groupoids'', that is, to groupoids equipped with an algebra structure for some endo-pseudofunctor on (univalent) groupoids.

Secondly, the univalence condition on categories was extended to bicategories in \cite{DBLP:journals/mscs/AhrensFMVW21} and to other (higher-)categorical structures in \cite{univalence_principle}.
In more detail, \cite{univalence_principle} develops a notion of theory for mathematical structures, and a notion of univalence for models of such theories.

Thirdly, univalent displayed graphs are used in \cite{univalent_graphs} to define and study higher groups.

In the present paper, we continue the study of univalent (higher-)categorical structures,
focusing on \emph{monoidal} categories.
Monoidal categories are very useful in a variety of contexts, such as 
quantum mechanics \cite{application_QFT} and computing \cite{application_quantum_computing}, 
modeling concurrency \cite{application_concurrency}, 
probability theory \cite{application_categorical_probability_theory} and probabilistic programming \cite{application_probabilistic_programming}, and
neural networks \cite{application_neuralnetworks}.
We present two results on monoidal categories:

\begin{enumerate}
\item We show that the bicategory of univalent monoidal categories is univalent. Here, a univalent monoidal category is a univalent category with a monoidal structure.
\item We construct, for any monoidal category, a monoidal Rezk completion. It is, in particular, a univalent monoidal category; the challenge lies in establishing the universal property of a Rezk completion, here modified for monoidal categories.
\end{enumerate}

Both results have been formalized in the \UniMath library of univalent mathematics, based on the Coq proof assistant.

The first of these results may be considered to be a basic sanity check; failing to prove this would question the validity of our definitions.
However, its proof is technically difficult, and, in our experience, only feasible through the disciplined application of ``displayed'' technology as developed in \cite{DBLP:journals/lmcs/AhrensL19} and \cite{DBLP:journals/mscs/AhrensFMVW21}.

The second result consists, more specifically, of a lifting of the Rezk completion for categories as constructed in \cite{DBLP:journals/mscs/AhrensKS15} to the monoidal structure.
As such, it also relies on displayed technology: the equivalence expressing the universal property of our monoidal Rezk completion is given as a displayed equivalence on top of the equivalence constructed in \cite{DBLP:journals/mscs/AhrensKS15}.

Our work is strongly related to some of the work mentioned above.

Firstly, an instance of Van der Weide's work covers monoidal groupoids; see \cite[Section~6.7.4]{phd/NVDW}.
Compared to that work, our work discusses monoidal \emph{categories} rather than groupoids, but does not cover general structures.
In particular, we also provide a completion operation for \emph{lax} and \emph{oplax} monoidal categories. 
Work on the ``pushout'' of our and Van der Weide's work, a Rezk completion for structured categories, is ongoing (see also \cref{sec:conclusion}).

Secondly, \cite[Example~8.7]{univalence_principle} studies monoidal categories. It is shown there that the general univalence condition on a model of the theory of monoidal categories defined in that work simplifies, in the case of monoidal categories, to the underlying category being univalent.
Thus, the univalent monoidal categories of \cite[Example~8.7]{univalence_principle} are the same as the ones studied in the present work.

In the remainder of the introduction, we review the Rezk completion and displayed (bi)categories, respectively.
We also give some details about the formalization.

\begin{notation}
In order to stay consistent with the notation used in \UniMath, we write the composition in diagrammatic order, \ie, the composition of $f: x\to y$ and $g:y\to z$ is denoted as $f\comp g: x \to z$.
\end{notation}

There are different notions of \textit{sameness} between categories:
\begin{definition} A functor $F : \CC \to \DD$ is called
\begin{enumerate}
\item a \textbf{weak equivalence} if it is fully faithful and essentially surjective;
\item a \textbf{(strong) equivalence} if it is fully faithful and split essentially surjective.
Equivalently, this means that $F$ is invertible up to a natural isomorphism; 
\item an \textbf{adjoint equivalence} is a (strong) equivalence $F$ whose inverse (up to a natural isomorphism) is the right adjoint of $F$;
\item an \textbf{isomorphism} if it is fully faithful and the function on objects is an equivalence of types.
\end{enumerate}
\end{definition}

Even though these four concepts are closely related, they enjoy different properties. 
The Rezk completion is, in general, only a weak equivalence; categorical structure does not necessarily transfer along a weak equivalence. 
For strict categories (\ie, categories whose type of objects is a set), the statement that every weak equivalence is an (adjoint) equivalence is equivalent to the axiom of choice.
However, if one restricts to univalent categories, these four notions are always equivalent (without using the axiom of choice). 

\subsection{Review of the Rezk completion for categories}
\label{sec:revi-rezk-compl}

The Rezk completion for categories was constructed in \cite{DBLP:journals/mscs/AhrensKS15}.
In essence, given a category $\C$, its Rezk completion is given by a univalent category $\RC(\C)$ and a weak equivalence $\HH : \C \to \RC(\C)$. %
This weak equivalence has the following property: any functor $F : \C \to \EE$, with $\EE$ a univalent category, factors \emph{uniquely} via $\HH$, as depicted in the following diagram.
\begin{equation}
  \begin{tikzcd}
    \C \ar[dr, "F"] \ar[d, "\HH"']
    \\
    \RC(\C) \ar[r, dashed, "\exists !"']
    &
    \EE
  \end{tikzcd}
  \label{eq:rezk-quotient}
\end{equation}

\begin{remark}
  The universal property satisfied by the Rezk completion is a bicategorical one, see \cref{def:left-universal-arrow}.
  From a purely category-theoretic viewpoint, the factorization in \cref{eq:rezk-quotient} is unique up to natural isomorphism.
  However, since $\EE$ is univalent, the functor category $[\RC(\C),\EE]$ is also univalent. 
Therefore, the factorization of such a functor is unique.
\end{remark}

In \cite{DBLP:journals/mscs/AhrensKS15}, it is said that the construction gives a universal way to replace a category by a univalent category. 
This construction is indeed universal in a bicategorical sense, according to the following lemma:
\begin{lemma}[{\cite[Thm.~8.4]{DBLP:journals/mscs/AhrensKS15}}, \coqident{CategoryTheory.PrecompEquivalence}{precomp_adjoint_equivalence}]
\label{weq-induces-iso-lemma}
Let $\HH : \CC \to \DD$ be a weak equivalence between categories. 
For any univalent category $\EE$, the functor 
$\precomp{\HH} : [\DD, \EE] \to [\CC,\EE]$
is an adjoint equivalence of categories.
\end{lemma}

\cref{weq-induces-iso-lemma}, when applied to the Rezk completion, provides an instance of a ``(left) universal arrow'':
\begin{definition}[\coqident{Bicategories.PseudoFunctors.UniversalArrow}{left_universal_arrow}]\label{def:left-universal-arrow}
Let $R : \BB_2 \to \BB_1$ be a pseudo-functor.
A \textbf{left universal arrow} from an object $x : \ob{(\BB_1)}$ to $R$ is given by:
\begin{enumerate}
\item an object $L\, x : \ob{(\BB_2)}$,
\item a morphism $\eta_x : \mor{\BB_1}{x}{R(L\,x)}$;
\item for any $y : \ob{(\BB_2)}$,
the functor 
\[
{\eta_x}\comp (R\,-) : \mor{\BB_2}{L\,x}{y} \to \mor{\BB_1}{x}{R\, y}\enspace,
\]
which acts on morphisms by applying $R$ and whiskering with $\eta_x$,
is an adjoint equivalence of categories.
\end{enumerate}
\end{definition}

\begin{remark}
Writing  $\CAT$ for the bicategory of categories, functors, and natural transformations,
and $\CATU$ for the full sub-bicategory of $\CAT$ consisting of \emph{univalent} categories, functors, and natural transformations,
\cref{weq-induces-iso-lemma} applied to the Rezk completion of $\CC$ provides a universal arrow from $\CC$ to the inclusion $\CATU \hookrightarrow \CAT$.
We expect the following to hold: if we have, for any object $x$, a left universal arrow to $R$ with object part $L\,x$, then the assignment $x \mapsto L\, x$ induces a pseudo-functor $L : \BB_1 \to \BB_2$ which is a left biadjoint to $R$. 
Hence, \cref{weq-induces-iso-lemma} applied to the Rezk completion would yield a left bi-adjoint to the inclusion $\CATU \hookrightarrow \CAT$.
However, we have not found a reference for the connection between universal arrows and biadjunctions.
As we do not need this correspondence, we do not develop it further.
\end{remark}

\begin{remark}
In \cite{DBLP:journals/mscs/AhrensKS15}, the Rezk completion has been constructed as the co-restriction of the Yoneda embedding to its image. 
It is already known how the Yoneda embedding transports the monoidal structure; more details on the connection between these approaches are given in \cref{sec:day-conv}. 
However, this construction raises the universe level of the type of objects and morphisms. 
In \url{https://1lab.dev/Cat.Univalent.Rezk.html}, the authors show how to decrease the universe level of the type of objects by one, using the construction of small images (and, in particular, higher inductive types).
One can also construct (the type of objects of) the Rezk completion as a higher inductive type. This has been done in \cite{hottbook}.

In this paper, we work with an \emph{abstract} Rezk completion of a category instead of a concrete implementation. 
Consequently, the approach presented here can be applied to any of those constructions.
\end{remark}

\subsection{Review of displayed (bi)categories}
\label{sec:revi-displ-bicat}
In this section, we recall the basic concepts of displayed bicategories and their univalence. 
More information can be found in \cite{DBLP:journals/mscs/AhrensFMVW21}.

Let us first briefly recall the idea of displayed categories.

\newcommand{\catb}[1]{\mathbf{#1}}
\newcommand{\SET}{\catb{Set}}
\newcommand{\MON}{\catb{Mon}}

Many concrete examples of categories are given by structured sets and structure-preserving functions.
An example of this is the category $\MON$ of monoids and monoid homomorphisms. 
In particular, an identity morphism is an identity function (\ie, the identity morphism in $\SET$) and the composition of monoid homomorphisms is given by the composition of the underlying functions (\ie, the composition in $\SET$).
Therefore, working in a category of structured sets often means lifting structure of the category $\SET$ to the additional structure. 
An example of this phenomenon is the product of monoids: the underlying set of a product of monoids can be constructed as the product of the underlying sets (\cref{category-groups-disp-example}).

The notion of \textbf{displayed category} formalizes the process of creating a new category out of an old category by adding structure and/or properties on the objects and/or morphisms in the following way:
a displayed category ({\cite[Def.~3.1]{DBLP:journals/lmcs/AhrensL19}}) specifies precisely the extra structure and the extra laws needed to build the new category out of the old one. This new category is then called the \emph{total category} of the displayed category ({\cite[Def.~3.2]{DBLP:journals/lmcs/AhrensL19}}). 

\begin{example}
\label{category-groups-disp-example}
The category $\MON$ of monoids can be constructed as a total category over $\SET$ as follows:
\begin{enumerate}
\item For $X : \SET$, the type of displayed objects over $X$ is the type of monoid structures on $X$:
\[
\sum_{m : X \times X \to X} \sum_{e : X} \cfont{isAssociative(}m\cfont{)} \times \prod_{x : X} \left(e\cdot x = x \times x \cdot e = x\right)\enspace,
\]
where $\cfont{isAssociative(}m\cfont{)}$ is the proposition stating that $m$ is associative.
\item Assume given $X,Y : \SET, f : \mor{\SET}{X}{Y}$ and $(m_X, e_X, p_X)$ (\resp $(m_Y, e_Y, p_Y)$) a displayed object over $X$ (\resp $Y$), \ie, the structure of a monoid. The type of displayed morphisms over~$f$ is the proposition stating that $f$ is a monoid homomorphism from $(m_X, e_X, p_X)$ to $(m_Y, e_Y, p_Y)$:
\[
\left(f\,e_X = e_Y\right) \times \prod_{x_1,x_2 : X} f\,(m_X(x_1,x_2)) = m_Y(f\,x_1,f\,x_2)\enspace.
\]
\end{enumerate}
\end{example}

Analogously, there is also the notion of a \textbf{displayed bicategory}:

\begin{definition}
[{\cite[Def.~6.1]{DBLP:journals/mscs/AhrensFMVW21}}, \coqident{Bicategories.DisplayedBicats.DispBicat}{disp_bicat}]
\label{disp-bicat-definition}
Let $\BB$ be a bicategory. A \textbf{displayed bicategory} $\DD$ over $\BB$ consists of:
\begin{enumerate}
\item for any $x : \BB$, a type $\DD_x$ of \emph{displayed objects} over $x$,
\item for any $f: \mor{\BB}{x}{y}$ and $\x : \DD_x$ and $\y : \DD_y$, a type $\DD_f(\x,\y)$ of \emph{displayed morphisms} over~$f$,
\item for any $\alpha : \mor{\BB}{x}{y}(f,g)$ and $\f : \DD_f(\x,\y)$ and $\g : \DD_g(\x,\y)$, a set $\f \xRightarrow{\alpha} \g$ of \emph{displayed $2$-cells} over $\alpha$;
\end{enumerate}
together with a composition of displayed morphisms and displayed $2$-cells (over the composition in $\BB$) and a displayed identity morphism and $2$-cell (over the identity morphism \resp $2$-cell in $\BB$).
The axioms of a bicategory have corresponding displayed axioms (over those axioms in $\BB$).
\end{definition}

\begin{definition} 
[{\cite[Def.~6.2]{DBLP:journals/mscs/AhrensFMVW21}}, \coqident{Bicategories.DisplayedBicats.DispBicat}{total_bicat}]
Let $\DD$ be a displayed bicategory over $\BB$. 
The \textbf{total bicategory} of $\DD$, denoted as $\int \DD$, has as $i$-cells (with $i=0,1,2$), pairs $(x,\x)$ where $x$ is an $i$-cell of $\BB$ and $\x$ is a displayed $i$-cell of $\DD$ over $x$.
\end{definition}

\begin{example}
The bicategory whose objects are categories equipped with a terminal object, whose morphisms are functors preserving the terminal objects (strongly) and whose $2$-cells are natural transformations,
can be constructed as a total bicategory over $\CAT$ as follows:
\begin{enumerate}
\item For $\CC : \CAT$, the type of displayed objects over $\CAT$ is the type expressing that $\CC$ has a terminal object:
\[
\sum_{X :\,\CC} \cfont{isTerminal(}X\cfont{)}\enspace.
\]
\item Assume given $\CC,\DD : \CAT, F : \mor{\CAT}{\CC}{\DD}$ and $(T_\CC, p_\CC)$ (\resp $(T_\DD,p_\DD)$) displayed objects over $\CC$ (\resp $\DD$).
The type of displayed morphisms over~$F$ is the proposition stating that $F$ preserves the terminal object:
\[
\cfont{isIsomorphism(}!\cfont{)}\enspace,
\]
where $!$ is the unique morphism $F\, T_\CC \to T_\DD$ given by the universal property of the terminal object $T_\DD$.
\item Let $F, G : \mor{\CAT}{\CC}{\DD}$ be functors between categories $\CC$ and $\DD$ and assume:
\begin{enumerate}
\item $(T_\CC, p_\CC)$ (\resp $(T_\DD,p_\DD)$) a witness that $\CC$ (\resp $\DD$) has a terminal object, \ie, it is a displayed object over $\CC$ (\resp $\DD$),
\item $\pt{F}$ (\resp $\pt{G}$) a proof witnessing that $F$ (\resp $G$) preserves the terminal object strongly, \ie, $\pt{F}$ (\resp $\pt{G}$) is a displayed morphism over $F$ (\resp $G$).
\end{enumerate}
For any natural transformation $\alpha : F \Rightarrow G$, the type of displayed $2$-cells over $\alpha$ is the unit type.
\end{enumerate}
\end{example}

Given displayed bicategories $\DD_1$ and $\DD_2$ over a bicategory $\BB$, 
we construct the product $\DD_1 \times \DD_2$ over $\BB$. 
The displayed objects, morphisms, and $2$-cells are pairs of objects, morphisms, and $2$-cells, respectively
(\coqident{Bicategories.DisplayedBicats.Examples.Prod}{disp_dirprod_bicat}).

A displayed bicategory is \emph{locally} \emph{univalent} if the function of type
\[
\bar{f} =_p \bar{g} \to \bar{f} \cong_{\mathsf{idtoiso}^{2,1}_{f,g}(p)} \bar{g}\enspace,
\]
sending $\mathsf{refl}$ to the identity displayed isomorphism,
is an equivalence of types for all morphisms $f$ and $g$ of the same type, $p : f = g$ and $\bar{f}$ (\resp $\bar{g}$) displayed morphisms over $f$ (\resp $g$).

A displayed bicategory is \emph{globally} \emph{univalent} if the function of type
\[
\bar{x} =_p \bar{y} \to \bar{x} \simeq_{\mathsf{idtoiso}^{2,0}_{x,y}(p)} \bar{y}\enspace,
\]
sending $\mathsf{refl}$ to the identity displayed adjoint equivalence,
is an equivalence of types for all objects $x$ and $y$, $p : x = y$ and $\bar{x}$ (\resp $\bar{y}$) displayed objects over $x$ (\resp $y$).

A displayed bicategory is \emph{univalent} if it is both locally and globally univalent
(\coqident{Bicategories.DisplayedBicats.DispUnivalence}{disp_univalent_2}, \coqident{Bicategories.DisplayedBicats.DispUnivalence}{disp_univalent_2_1}, \coqident{Bicategories.DisplayedBicats.DispUnivalence}{disp_univalent_2_0}).

\begin{lemma}
[{\cite[Thm.~7.4]{DBLP:journals/mscs/AhrensFMVW21}}, \coqident{Bicategories.DisplayedBicats.DispUnivalence}{total_is_univalent_2}]
\label{total-bicat-univalent-lemma}
Let $\DD$ be a displayed bicategory over $\BB$ and $q\in\{\textit{locally},\textit{globally}\}$. Then
$\int D$ is $q$-univalent if $\BB$ is $q$-univalent and $\DD$ is $q$-univalent.
\end{lemma}

\begin{remark}
As witnessed by \cref{total-bicat-univalent-lemma}, certain properties of the total bicategory can be expressed in terms of the \emph{base} bicategory and the displayed bicategory.
This allows one to divide a problem, in this case showing univalence, into multiple steps.

Therefore, while we are interested in studying the total bicategory, we usually only describe the displayed bicategory.
\end{remark}

\begin{definition}
[{\cite[Def.~7.7]{DBLP:journals/mscs/AhrensFMVW21}}, \coqident{Bicategories.DisplayedBicats.DispBicat}{disp_locally_groupoid}, {\cite[Def.~7.8]{DBLP:journals/mscs/AhrensFMVW21}}, \coqident{Bicategories.DisplayedBicats.DispBicat}{disp_2cells_isaprop}]
A displayed bicategory $\DD$, over a bicategory $\BB$, is called
\begin{enumerate}
\item \textbf{Locally groupoidal} if all displayed $2$-cells over invertible $2$-cells are invertible;
\item \textbf{Locally propositional} if each type of displayed $2$-cells is a proposition.
\end{enumerate}
\end{definition}

We will also need the displayed analogue of the concept of a functor being essentially surjective:
\begin{definition}
[\coqident{CategoryTheory.DisplayedCats.Functors}{disp_functor_disp_ess_split_surj}]
A displayed functor $\bar{F} : \DD_1 \to \DD_2$ over a functor $F : \CC_1 \to \CC_2$ is \textbf{displayed split essentially surjective} if for any $x : \CC$ and $\bar{y} : (\DD_2)_{F\, x}$, 
a displayed object $\bar{x} : (\DD_1)_x$ is given together with a displayed isomorphism between $\bar{F}\, \bar{x}$ and $\bar{y}$ over the identity isomorphism $\id{F\,x}$.
\end{definition}

\subsection{Formalization in UniMath}
\label{sec:form-unim}

The results presented here are formulated inside intensional dependent type theory. 
We carefully distinguish between data and properties, \ie, data is always explicitly given which avoids the use of the axiom of choice and the law of excluded middle.
The results presented here are formalized and checked in the library \UniMath \cite{UniMath} of univalent mathematics, based on the proof assistant \Coq \cite{coq}.

The formalization referred to in this paper is presented in the \UniMath commit \href{https://github.com/UniMath/UniMath/tree/\longhash}{\texttt{\shorthash}} (more precisely, the given link leads to the source code repository right after merging this commit).
A generated HTML documentation of the sources at this commit is hosted online.
Most of our definitions, lemmas, and theorems are accompanied by a link which leads to the corresponding definition, lemma, and theorem in the documentation.

The formalization is built upon the existing library of (bi)category theory and the theory of displayed (bi)categories. The ($1$-)categorical formulation of displayed categories has been developed in \cite{DBLP:journals/lmcs/AhrensL19} and the bicategorical formulation has been developed in \cite{DBLP:journals/mscs/AhrensFMVW21}.

The accompanying code, specific to this work, consists of approximately $7000$ lines of code. However, the formalisation also made it necessary to contribute to the \UniMath library on monoidal categories more generally.

\section{The bicategory of monoidal categories}
\label{sec:bicat-mono-cat-constr}
In this section we construct the bicategory $\MONCAT$ (\resp $\MONCATS$) of monoidal categories, lax (\resp strong) monoidal functors and monoidal natural transformations.
We construct this bicategory as the total bicategory of a displayed bicategory over the bicategory $\CAT$ of categories, functors, and natural transformations. 

This displayed bicategory in itself is constructed by stacking different displayed bicategories. 
First, we construct a displayed bicategory $\CATT$ (\resp $\CATUnit$) over $\CAT$ that adds a tensor (\resp a unit).
Then, we construct displayed bicategories $\CATLU,\CATRU$ and $\CATA$ over the total bicategory of $\CATTU := \CATT \times \CATUnit$ that add the left unitor, right unitor and the associator, respectively. The product of these displayed bicategories is denoted by $\CATUA$ and the laws that relate the unitors and the associator, \eg, the triangle and pentagon identities, are represented by a full (displayed) sub-bicategory $\CATP$ of $\CATUA$. Lastly, we also have a displayed (sub)bicategory $\CATS$ of $\CATP$ that enforces the strongness of the monoidal functors.

The construction is summarized in \cref{fig:constructionmoncat}.
\begin{figure}[tb]
\begin{equation}
\label{img:diagram_constr_dispcats}
\begin{tikzcd}
& \CATS \arrow[d, "Def. \ref{disp-bicat-monstrong-definition}", hook] & \\
& \CATP \arrow[d, "Def. \ref{disp-bicat-mon-definition}" , hook] & \\
& \CATUA \arrow[ld, dashed] \arrow[d, dashed] \arrow[rd, dashed] & \\
\CATLU \arrow[rd, swap, "Def. \ref{disp-bicat-lunit-definition}"] & \CATRU \arrow[d,swap, "Def. \ref{disp-bicat-runit-definition}"] & \CATA \arrow[ld, "Def. \ref{disp-bicat-ass-definition}"] \\
& \CATTU \arrow[ld, dashed] \arrow[rd, dashed] & \\
\CATUnit \arrow[r,swap, "Def. \ref{disp-bicat-unit-definition}"] & \CAT & \CATT \arrow[l, "Def. \ref{disp-bicat-tensor-definition}"]
\end{tikzcd}
\end{equation}
\caption{Overview of construction steps towards $\MONCAT$ and $\MONCATS$}\label{fig:constructionmoncat}
\end{figure}
The precise meaning of this diagram is explained in the rest of this section and further explained in \cref{rem:diagram_constr_dispcats}.

\begin{remark}
Although the construction of $\MONCAT$ (\resp $\MONCATS$) is standard (when working in univalent foundations), we explain the construction in quite some detail because  both \cref{sec:bicat-mono-categ} and \cref{sec:rezk-compl-mono} heavily depend on the construction of monoidal categories (\resp lax/strong monoidal functors and natural transformations) in this displayed way. 
In particular, this allows us to fix notation and allows for the big picture of the constructions to become more visible. 
\end{remark}

The first displayed bicategory we construct adds the structure of a tensor and a unit. 
Since the unit and tensor are (without the unitors) independent of each other, we can define this as the product of displayed bicategories, the first representing the tensor and the second representing the unit.

\begin{definition}
[\coqident{Bicategories.MonoidalCategories.UnivalenceMonCat.TensorLayer}{bidisp_tensor_disp_bicat}]
\label{disp-bicat-tensor-definition}
The displayed bicategory $\CATT$ over $\CAT$ is defined as follows:

\begin{enumerate}
\item The displayed objects over a category $\CC : \CAT$ are the functors of type $\CC \times \CC \to \CC$, called \emph{tensors} over $\CC$ and are denoted by $\tensor_\CC$.

\item The displayed morphisms over a functor $F:\CC\to\DD$ from $\tensor_\CC$ to $\tensor_D$
are the natural transformations of type $(F\times F) \comp \tensor_\DD \Rightarrow \tensor_\CC \comp F$, called \emph{witnesses of tensor-preservation of $F$} and are denoted by $\pt{F}$.
\item The displayed $2$-cells over a natural transformation $\alpha : F\Rightarrow G$ from  $\pt{F}$ to $\pt{G}$ are the proofs of the proposition
\[
\prod_{x,y : \CC} (\alpha_x \tensor_D \alpha_y) \comp \pt{G}_{x, y}
                  = \pt{F}_{x,y} \comp \alpha_{x \tensor_C y}\enspace.
\]

\end{enumerate}
\end{definition}

\begin{definition}
[\coqident{Bicategories.MonoidalCategories.UnivalenceMonCat.UnitLayer}{bidisp_unit_disp_bicat}]
\label{disp-bicat-unit-definition}
The displayed bicategory $\CATUnit$ over $\CAT$ is defined such that:

\begin{enumerate}
\item The displayed objects over a category $\CC : \CAT$ are the objects of $\CC$, called \emph{units} over $\CC$ and are denoted by $\unit_\CC$.

\item The displayed morphisms over a functor $F:\CC\to\DD$ from $\unit_\CC$ to $\unit_\DD$
are the morphisms of type $\mor{\DD}{\unit_\DD}{F \unit_\CC}$, called \emph{witnesses of unit-preservation of $F$} and are denoted by $\pu{F}$.

\item The displayed $2$-cells over a natural transformation $\alpha : F\Rightarrow G$ from  $\pu{F}$ to $\pu{G}$ are the proofs of the proposition \[ \pu{F} \comp \alpha_{\unit_\CC} = \pu{G}\enspace. \]

\end{enumerate}
\end{definition}

We denote by $\CATTU$ the displayed bicategory which is the product of $\CATT$ and $\CATUnit$ (\coqident{Bicategories.MonoidalCategories.UnivalenceMonCat.TensorUnitLayer}{bidisp_tensor_unit}).

To fix some notation: 
The total bicategory $\int \CATTU$ has as objects triples $(\CC,\tensor_{\CC},\unit_{\CC})$ where $\CC$ is a category, $\tensor_\CC$ a tensor on $\CC$ and $\unit_\CC$ a unit on $\CC$. 
A morphism from $(\CC,\tensor_\CC,\unit_\CC)$ to $(\DD,\tensor_\DD,\unit_\DD)$ is a triple $(F,\pt{F},\pu{F})$ where $F$ is a functor of type $\CC \to \DD$, $\pt{F}$ a witness of tensor-preservation of $F$ and $\pu{F}$ a witness of unit-preservation of $F$.

We now add the unitors and the associator. Since they are independent of each other (before adding the triangle and pentagon equalities), we can again define them as a product of displayed bicategories.
These displayed bicategories have trivial displayed $2$-cells since monoidal natural transformations only use the data of the tensor and the unit.
Thus we define these displayed bicategories as displayed categories. 
The formal construction of turning a displayed category into a displayed bicategory with trivial $2$-cells is formalized as \coqident{Bicategories.DisplayedBicats.Examples.DisplayedCatToBicat}{disp_cell_unit_bicat}.

\begin{definition}
[\coqident{Bicategories.MonoidalCategories.UnivalenceMonCat.AssociatorUnitorsLayer}{bidisp_lu_disp_bicat}]
\label{disp-bicat-lunit-definition}
The displayed bicategory $\CATLU$ over $\int \CATTU$ is defined as the displayed category (with trivial $2$-cells) such that:

\begin{enumerate}
\item The displayed objects over a triple $(\CC,\tensor_C,\unit_C)$ are the natural isomorphisms of type ${(\unit_\CC \tensor_\CC -)}\Rightarrow{\id{\CC}}$, called \emph{left unitors} over $(\CC,\tensor_C,\unit_C)$ and are denoted by $\lu^\CC$.

\item The displayed morphisms over a triple $(F,\pt{F},\pu{F})$ from $\lu^\CC$ to $\lu^\DD$ are proofs of the proposition:
\[
\prod_{x : \CC} (\pu{F} \tensor_\DD \id{F x}) 
	\comp \pt{F}_{\unit_\CC, x} \comp F \lu^\CC_{x} = \lu^\DD_{F x}\enspace.
\]
\label{disp-bicat-lunit-definition-morphisms}
\end{enumerate}
\end{definition}

\begin{definition}
[\coqident{Bicategories.MonoidalCategories.UnivalenceMonCat.AssociatorUnitorsLayer}{bidisp_ru_disp_bicat}]
\label{disp-bicat-runit-definition}
The displayed bicategory $\CATRU$ over $\int \CATTU$ is defined as the displayed category (with trivial $2$-cells) such that:

\begin{enumerate}
\item The displayed objects over a triple $(\CC,\tensor_C,\unit_C)$ are the natural isomorphisms of type ${(- \tensor_\CC \unit_\CC)}\Rightarrow {\id{\CC}}$, called \emph{right unitors} over $(\CC,\tensor_C,\unit_C)$ and are denoted as $\ru^\CC$.

\item The displayed morphisms over a triple $(F,\pt{F},\pu{F})$ from $\ru^\CC$ to $\ru^\DD$ are proofs of the proposition:
\[
\prod_{x : \CC} (\id{F x} \tensor_\DD \pu{F}) 
	\comp \pt{F}_{x,\unit_\CC} \comp F \ru^\CC_{x} = \ru^\DD_{F x}\enspace.
\]

\end{enumerate}
\end{definition}

\begin{definition}
[{\coqident[bidisp_associator_disp_bicat]{Bicategories.MonoidalCategories.UnivalenceMonCat.AssociatorUnitorsLayer}{bidisp_ass_disp_bicat}}]
\label{disp-bicat-ass-definition}
The displayed bicategory $\CATA$ over $\int \CATTU$ is defined as the displayed category (with trivial $2$-cells) such that:

\begin{enumerate}
\item The displayed objects over a triple $(\CC,\tensor_C,\unit_C)$ are the natural isomorphisms of type ${((- \tensor_\CC -) \tensor_\CC -)}\Rightarrow{(- \tensor_\CC (- \tensor_\CC -))}$, called \emph{associators} over $(\CC,\tensor_C,\unit_C)$ and are denoted as $\ass^\CC$.

\item The displayed morphisms over a triple $(F,\pt{F},\pu{F})$ from $\ass^\CC$ to $\ass^\DD$ are proofs of the proposition:
\[
\prod_{x,y,z: \CC} (\pt{F}_{x,y} \tensor_\DD \id{F z})
	\comp \pt{F}_{x \tensor_\CC y,z}
     \comp F \ass^{\CC}_{x,y,z}
	      = \ass^{\DD}_{F x, F y, F z}
		\comp (\id{F x} \tensor_\DD \pt{F}_{y,z})
		\comp \pt{F}_{x,y \tensor_\CC z}\enspace.
\]

\end{enumerate}
\end{definition}

We denote by $\CATUA$ the displayed bicategory over $\int\CATTU$ which is the product of $\CATLU, \CATRU$ and $\CATA$ (\coqident{Bicategories.MonoidalCategories.UnivalenceMonCat.AssociatorUnitorsLayer}{bidisp_assunitors_disp_bicat}).

\begin{definition}
[\coqident{Bicategories.MonoidalCategories.UnivalenceMonCat.FinalLayer}{disp_bicat_univmon}]
\label{disp-bicat-mon-definition}
The displayed bicategory $\CATP$ is the full displayed sub-bicategory of $\CATUA$ specified by the product of the following predicates:
\begin{enumerate}
\item Triangle equality:
\[
\prod_{x,y:\CC} \ass_{x,\unit,y} \comp \id{x} \tensor \lu_y = \ru_x \tensor \id{y}\enspace.
\]
\item Pentagon equality:
\[
\prod_{w,x,y,z : \CC}
      (\ass_{w,x,y} \tensor \id{z}) 
      \comp \ass_{w,x \tensor y, z}
      \comp \id{w} \tensor \ass_{x,y,z}
      = \ass_{w \tensor x,y,z} \comp \ass_{w,x,y \tensor z}\enspace.
\]
\end{enumerate}
\end{definition}

\begin{definition}
[\coqident{Bicategories.MonoidalCategories.UnivalenceMonCat.FinalLayer}{disp_bicat_univstrongfunctor}]
\label{disp-bicat-monstrong-definition}
The displayed bicategory $\CATS$ is the (non-full) displayed sub-bicategory of $\CATP$ where the displayed morphisms are proofs of the proposition
\[
\isiso{\pu{}} \times \prod_{x,y:\CC} \isiso{\pt{}_{x,y}}\enspace.
\]
\end{definition}

The bicategory of monoidal categories, lax (\resp strong) monoidal functors, and monoidal natural transformations is denoted by $\MONCAT := \int \CATP$ (\resp $\MONCATS := \int \CATS$).

\begin{remark}
\label{rem:diagram_constr_dispcats}
The constructions are summarized in Figure \ref{img:diagram_constr_dispcats}.
The dashed arrows correspond to the projection induced by the product of the displayed bicategories to any of the components. 
In particular, this means that the dashed arrows induce a (bi)pullback (of displayed bicategories). 
The filled arrows represent that we have a forgetful pseudofunctor (given by the projection of a total bicategory to its base bicategory).
Lastly, the hooked arrows mean that the domain is constructed as a (displayed) full sub-bicategory.
\end{remark}

\begin{remark}
\label{moncat-sigma-construction-remark}
An object in $\MONCAT$ is of the form $(((\CC,\tensor,\unit),\lu,\ru,\ass),tri,pent)$.
Usually, one wants to consider an object in $\MONCAT$ to be of the form $(\CC, (((\tensor,\unit), \lu,\ru,\ass), tri,pent))$, \ie, as a category equipped with a monoidal structure.
The displayed bicategory whose objects are categories equipped with a monoidal structure can be constructed by applying the sigma construction (\cite[Definition 6.6(2)]{DBLP:journals/mscs/AhrensFMVW21},\coqident{Bicategories.DisplayedBicats.Examples.Sigma}{sigma_bicat}).
Furthermore, this displayed bicategory is univalent by a criterion presented in \cite{DBLP:journals/mscs/AhrensFMVW21}.
As this does not change the message of the paper, we refer the reader to \cite{DBLP:journals/mscs/AhrensFMVW21} for the precise statements, but we do show that the criteria are satisfied in  \cref{lemma:univCATT_groupoidal,lemma:univCATUnit_groupoidal,lemma:univCATUA_groupoidal}.
\end{remark}

\begin{remark}
In the formalization of $\CATLU$ (\resp $\CATRU$, $\CATA$), we do not yet require a left unitor (\resp right unitor, associator) to be an isomorphism. Since being an isomorphism is a proposition, we could and did add these three (indexed) conditions only in the formalization of $\CATP$.
This simplifies the proof of univalence of the bicategory of univalent monoidal categories that is built from $\MONCAT$.
\end{remark}

In \cref{sec:rezk-compl-mono}, we construct a Rezk completion for monoidal categories. 
We are interested in studying the hom-categories of $\MONCAT$ and thus, in particular, the displayed hom-categories.
We now introduce some notations.
Let $\BB$ be a bicategory and $x,y : \BB$ objects. The hom-category from $x$ to $y$ is denoted by $\mor{\BB}{x}{y}$.
Any morphism $f : \mor{\BB}{x}{y}$ induces a functor between hom-categories, more precisely: 
\begin{definition}
\label{precomp-functor-definition}
Let $\BB$ be a bicategory, $f : \mor{\BB}{x}{y}$ a morphism and $z : \BB$ an object.
The \textbf{functor given by precomposition with $f$ and target object $z$} is the functor
\[
\precomp{f} : \mor{\BB}{y}{z} \to \mor{\BB}{x}{z}\enspace,
\]
where the action on the objects is given by precomposition, \ie, $g \mapsto f\cdot g$, 
and the action on the morphisms is given by left whiskering, \ie, $\alpha\mapsto f \lwhisker \alpha$.
\end{definition}
We also refer to the functor given by precomposition with $f$ as the \textbf{precomposition functor with $f$}.

Let $\DD$ be a displayed bicategory over $\BB$ and $\x \in \DD_x$ and $\y \in \DD_y$ be displayed objects.
The (total) hom-category $\mor{\int \DD}{(x,\x)}{(y,\y)}$ can be constructed as a total category of a displayed category over $\mor{\BB}{x}{y}$. %
We denote this displayed category by $\mor{\DD}{\x}{\y}$ (so we use the same notation for the hom-categories and displayed hom-categories).

In particular, the precomposition functor \wrt the total bicategory $\int \DD$ of a morphism $(f,\f)$ can be defined as a displayed functor over the precomposition functor $\precomp{f}$ (\wrt $\BB$) where we precompose/left whisker (in the displayed sense) with $\f$:
\begin{definition}
\label{disp-precomp-functor-definition}
Let $\DD$ be a displayed bicategory over a bicategory $\BB$, $\x : \DD_x, \y : \DD_y$ displayed objects, $\f : \DD_f(\x,\y)$ a displayed morphism and $\z : \DD_z$ a displayed object.
The \textbf{displayed functor given by precomposition with $\f$ and target displayed object $\z$} is the displayed functor
\[
\precomp{\f} : \mor{\DD}{\y}{\z} \to \mor{\DD}{\x}{\z}\enspace
\]
over the functor given by precomposition with $f$ and target object $z$.
\end{definition}
We also refer to the displayed functor given by precomposition with $\f$ as the \textbf{displayed precomposition functor with $\f$}.

\section{The univalent bicategory of monoidal categories}
\label{sec:bicat-mono-categ}

In this \lcnamecref{sec:bicat-mono-categ} we present our proof of univalence of the bicategory $\MONCATU$ of univalent monoidal categories, with \cref{umoncat-univ-theorem} as the main result. 
(We also obtain a version with strong monoidal functors in place of lax monoidal functors.)
In this proof, we rely heavily on the \emph{displayed} machinery built in \cite{DBLP:journals/mscs/AhrensFMVW21}, for modular construction of bicategories, and proofs of their univalence.

In the formalization of this univalence proof, we have not used the formalization of a monoidal category as presented above. Instead, we have changed the definition of a tensor from being a functor to a more explicit, unfolded definition.
It is not necessarily obvious that the resulting bicategory is indeed that of monoidal categories, lax (\resp strong) monoidal functors, and monoidal natural transformations. Therefore, we construct an equivalence of types of monoidal categories as presented above on the one hand and using this explicit definition on the other hand (\coqident{Bicategories.MonoidalCategories.UnivalenceMonCat.EquivalenceMonCatNonCurried}{cmonoidal_to_noncurriedmonoidal}, \coqident{Bicategories.MonoidalCategories.UnivalenceMonCat.EquivalenceMonCatNonCurried}{cmonoidal_adjequiv_noncurried_hom}).

Recall from \cref{total-bicat-univalent-lemma} that the total bicategory of a displayed bicategory is univalent if both
the base bicategory and the displayed bicategory are univalent.
Since $\CATU$ is univalent [{\cite[Prop.~3.19]{DBLP:journals/mscs/AhrensFMVW21}}, \coqident{Bicategories.Core.Examples.BicatOfUnivCats}{univalent_cat_is_univalent_2}], the task of proving $\MONCATU$ univalent therefore reduces to showing that $\Sigma_{\Sigma_{\CATTU} \CATUA} \CATP$ from the previous section is univalent, restricted to the full sub-bicategory $\CATU$ of $\CAT$. (This is to be read modulo the repackaging hinted to in \cref{moncat-sigma-construction-remark}.)

The sigma construction of univalent displayed bicategories is univalent provided that both displayed bicategories are locally groupoidal and locally propositional [{\cite[Prop.~7.9]{DBLP:journals/mscs/AhrensFMVW21}}, \coqident{Bicategories.DisplayedBicats.Examples.Sigma}{sigma_disp_univalent_2_with_props}].
The previously defined displayed bicategories are locally propositional since they either express an (indexed) equality of morphisms or the type of $2$-cells is the unit type.
Thus in this section, we show that the displayed bicategories from \cref{sec:bicat-mono-cat-constr} are univalent and locally groupoidal.

\begin{remark}
In this section we restrict the displayed bicategories to the bicategory $\CATU$ of univalent categories. For example, the restriction of $\CATTU$ is considered as the pullback of the displayed bicategory $\CATTU$ along the inclusion of $\CATU$ into $\CAT$.
We denote the restriction of the displayed bicategory $\CATl$ by $\univCATl$ for $\ell\in\{T, U, TU, LU, RU, A, \mathit{UA}, P, S\}$.
\end{remark}

\begin{lemma}
[{\coqident[tensor_disp_is_univalent_2]{Bicategories.MonoidalCategories.UnivalenceMonCat.TensorLayer}{bidisp_tensor_disp_prebicat_is_univalent_2}}]
	$\univCATT$ is univalent.
\end{lemma}
\begin{proof}
	$\univCATT$ is locally univalent by a straightforward calculation, we therefore only discuss that it is globally univalent.

	Let $\tensor_1, \tensor_2$ be two tensors on $\CC$. 
	We have to show that $\idtoisoob{\tensor_1}{\tensor_2}$ is an equivalence of types. 
	In order to show this, we factorize this function as follows:
	\[
	\begin{tikzcd}
		{\tensor_1 = \tensor_2} 
			\arrow[d,swap, "\idtoeq"] 
			\arrow[rr, "{\idtoisoob{\tensor_1}{\tensor_2}}"]
		&& {\mathsf{DispAdjEquiv}(\tensor_1,\tensor_2)} \\
		{\tensorEq(\tensor_1,\tensor_2)} \arrow [rr, "\eqtoiso", swap]
		&& {\tensorIso(\tensor_1,\tensor_2)} \arrow[u]
	\end{tikzcd},
	\]	
	where \(\tensorEq(\tensor_1,\tensor_2)\) is the type 
	\[
	\sum_{\alpha : {\prod}_{x,y : \CC}, x \tensor_1 y = x \tensor_2 y}
		\prod_{f : \mor{\CC}{x_1}{x_2}} \prod_{g : \mor{\CC}{y_1}{y_2}}
		f\tensor_1 g = f \tensor_2 g\enspace,
	\]
	where the equality $f\tensor_1 g = f \tensor_2 g$ is dependent over $\alpha_{x_1,y_1}$ and $\alpha_{x_2,y_2}$.
	
	The type \(\tensorIso(\tensor_1,\tensor_2)\) is the same as \(\tensorEq(\tensor_1,\tensor_2)\) where we replaced the first equality by an isomorphism (and the dependent equality of morphisms is replaced by pre- and post-composing with the isomorphism).
	
	The function $\idtoeq : \tensor_1 = \tensor_2 \to \tensorEq(\tensor_1,\tensor_2)$ maps equality to pointwise equality (on both the objects and morphisms). Because our hom-types are sets, this is an equivalence. 
	The function $\eqtoiso : \tensorEq(\tensor_1,\tensor_2) \to \tensorIso(\tensor_1,\tensor_2)$ replaces identity by isomorphism. Since $\CC$ is a univalent category, $\eqtoiso$ is indeed an equivalence. 
	Since a displayed adjoint equivalence in $\CATT$ translates into the notion of \(\tensorIso(\tensor_1,\tensor_2)\), 
	we construct in a straightforward manner a function from \(\tensorIso(\tensor_1,\tensor_2)\) to $\mathsf{DispAdjEquiv}(\tensor_1,\tensor_2)$,
	which is for the same reason an equivalence.
\end{proof}

Each type of (displayed) $2$-cells in $\CATUnit$ is contractible, hence:
\begin{lemma}
[{\coqident[tensor_disp_locally_groupoidal]{Bicategories.MonoidalCategories.UnivalenceMonCat.TensorLayer}{bidisp_tensor_disp_locally_groupoid}}]
\label{lemma:univCATT_groupoidal}
$\univCATT$ is locally groupoidal.
\end{lemma}
\begin{proof}
$\univCATT$ being locally groupoidal means that if a natural isomorphism $\alpha$ preserves the tensor, then so does its inverse.
This is immediate since the tensor product of isomorphisms is again an isomorphism (by functoriality of the tensor).
\end{proof}

\begin{lemma}
[{\coqident[unit_disp_is_univalent_2]{Bicategories.MonoidalCategories.UnivalenceMonCat.UnitLayer}{bidisp_unit_disp_prebicat_is_univalent_2}}]
	$\univCATUnit$ is univalent.
\end{lemma}
\begin{proof}
	$\univCATUnit$ is locally univalent by a straightforward calculation.
	Therefore, we only discuss why it is globally univalent.

	Let $I,J : \CC$ be objects representing a unit object. 
	As with the tensor layer, we factorize ${\idtoisoob{I}{J}}$ and show that each function in the factorization is an equivalence. The factorization is given by:	
	\[
	\begin{tikzcd}
		{I = J} 
			\arrow[dr] 
			\arrow[rr, "{\idtoisoob{I}{J}}"]
		&& {\mathsf{DispAdjEquiv}(I,J)} \\
		& {I\cong J}
			\arrow[ru] &		
	\end{tikzcd}
	\]
	The definition of a displayed adjoint equivalence in this displayed bicategory translates precisely to an isomorphism in the underlying category $\CC$, which gives us the arrow to the right and a proof that it is an equivalence. The left arrow is given by $\idtoiso_{I,J}$ and is an equivalence precisely because $\CC$ is a univalent category.
\end{proof}

\begin{lemma}
[{\coqident[unit_disp_locally_groupoidal]{Bicategories.MonoidalCategories.UnivalenceMonCat.UnitLayer}{bidisp_unit_disp_locally_groupoid}}]
\label{lemma:univCATUnit_groupoidal}
$\univCATUnit$ is locally groupoidal.
\end{lemma}

\begin{lemma}
[{\coqident[assunitors_disp_is_univalent_2]{Bicategories.MonoidalCategories.UnivalenceMonCat.AssociatorUnitorsLayer}{bidisp_assunitors_is_disp_univalent_2}}]
$\univCATUA$ is univalent.
\end{lemma}
\begin{proof}

Since the product of univalent displayed bicategories is univalent, it remains to show that $\univCATLU$,$\univCATRU$ and $\univCATA$ are univalent.

These displayed bicategories are locally univalent because the type of (displayed) $2$-cells is the unit type and the type of (displayed) $1$-cells is a proposition.

Since the type of objects (\resp morphisms, 2-cells) is a set (\resp proposition, contractible) and the base category is locally univalent, we can apply \cite[Prop.~7.10]{DBLP:journals/mscs/AhrensFMVW21}.
This proposition asserts that a displayed bicategory is univalent if a function of type $(a \simeq_{\idtoiso^{2,0}(p)} b) \to (a =_p b)$ can be constructed.
The latter means precisely that we have to construct displayed morphisms over an identity morphism.
In the case of the left unitor, this means that we have to construct a term of type $(\lu_1 = \lu_2)$ provided that the identity morphism on $(\CC,\tensor,\unit)$ preserves the left unitor (as in \cref{disp-bicat-lunit-definition}.\ref{disp-bicat-lunit-definition-morphisms}).
The proofs that $\univCATRU$ and $\univCATA$ are univalent is analogous.
\end{proof}

\begin{lemma}
[{\coqident[assunitors_disp_locally_groupoidal]{Bicategories.MonoidalCategories.UnivalenceMonCat.AssociatorUnitorsLayer}{bidisp_assunitors_disp_locally_groupoid}}]
\label{lemma:univCATUA_groupoidal}
$\univCATUA$ is locally groupoidal.
\end{lemma}
\begin{proof}
This follows from the following lemmas:
\begin{enumerate}
\item The product of locally groupoidal displayed bicategories is locally groupoidal.
\item A displayed bicategory whose type of displayed $2$-cells is the unit is locally groupoidal.
\end{enumerate}
\end{proof}

A full displayed sub-bicategory of a univalent displayed bicategory is univalent, hence:
\begin{lemma}
[{\coqident[tripent_disp_is_univalent_2]{Bicategories.MonoidalCategories.UnivalenceMonCat.FinalLayer}{disp_bicat_tripent_is_univalent_2}}]
\label{ucatp-univ-lemma}
$\univCATP$ is univalent.
\end{lemma}

Since a full displayed sub-bicategory of a displayed locally groupoidal bicategory is locally groupoidal, we have that $\univCATP$ is locally groupoidal.

\begin{theorem}
[\coqident{Bicategories.MonoidalCategories.UnivalenceMonCat.FinalLayer}{UMONCAT_is_univalent_2}]
\label{umoncat-univ-theorem}
The bicategory of univalent monoidal categories, lax monoidal functors, and monoidal natural transformations is univalent.
\end{theorem}

\begin{lemma}
[{\coqident[UMONCAT_disp_strong_is_univalent_2]{Bicategories.MonoidalCategories.UnivalenceMonCat.FinalLayer}{disp_bicat_univstrongfunctor_is_univalent_2}}]
$\univCATS$ is univalent.
\end{lemma}
\begin{proof}
This follows immediately from \cref{ucatp-univ-lemma} since the type of displayed $1$-cells is a mere proposition.
\end{proof}

\begin{theorem}
[\coqident{Bicategories.MonoidalCategories.UnivalenceMonCat.FinalLayer}{UMONCAT_strong_is_univalent_2}]
The bicategory of univalent monoidal categories, strong monoidal functors, and monoidal natural transformations is univalent.
\end{theorem}

\section{The Rezk completion for monoidal categories}
\label{sec:rezk-compl-mono}

Some constructions of (monoidal) categories do not yield univalent (monoidal) categories.
For instance, categories built from syntax usually have \emph{sets} of objects; the presence of non-trivial isomorphisms in such a category hence entails that it is not univalent.
Another example is when constructing colimits of univalent monoidal categories; the usual construction of such a colimit often yields a non-univalent monoidal category.
In such cases, a ``completion operation'', turning a monoidal category into a univalent one, is handy.

In this \lcnamecref{sec:rezk-compl-mono} we construct, for each monoidal category, a free univalent monoidal category, which we call the \emph{monoidal Rezk completion}. 
More precisely, we solve the following problem: 

\begin{problem}\label{prop:monoidal_rezk}
Given a Rezk completion $\HH:\CC\to\DD$ of a category $\CC$ and a monoidal structure $M := (\tensor, \unit, \lu,\ru, \ass)$ on $\CC$,
construct a monoidal structure $\liftstruct{M} := (\tensorD, \unitD, \luD,\ruD, \assD)$ on $\DD$
and a strong monoidal structure for $\HH$ \wrt $M$ and $\liftstruct{M}$,
such that
for any univalent monoidal category $(\EE,N)$,
the isomorphism of categories
\[
\precomp{\HH} : \mor{\CAT}{\DD}{\EE} \to \mor{\CAT}{\CC}{\EE}
\]
lifts to the category of lax (\resp strong) monoidal functors:
\[
\precomp{\HH} : \mor{\MONCAT}{(\DD,\liftstruct{M})}{(\EE,N)} \to \mor{\MONCAT}{(\CC,M)}{(\EE,N)}\enspace.
\]
\end{problem}

Once solved, we call $(\DD, \liftstruct{M})$ the \emph{monoidal Rezk completion of $(\CC, M)$}.
Analogous to the Rezk completion for categories, the monoidal Rezk completion exhibits the bicategory $\MONCATU$ (\resp $\MONCATSU$) as a reflective full sub-bicategory of $\MONCAT$ (\resp $\MONCATS$).

Although any categorical structure on a category can be transported along an equivalence of categories such that they become equivalent in the corresponding bicategory of structured categories, this might not be the case if one considers a weak equivalence.
On the way towards solving \cref{prop:monoidal_rezk}, we show, in particular, how to transport a monoidal structure along a weak equivalence of categories (see \cref{dfn:weq_transport_mon}), provided that the target category is univalent.
That construction is not limited to the specific weak equivalence given by the Rezk completion.

Analogous to the univalence proof of $\MONCATU$ (\resp $\MONCATSU$) given in \cref{sec:bicat-mono-categ}, we rely on the theory of displayed categories in order to solve this problem by dividing it into subgoals.
In each of the subgoals, we use the same strategy. 
In \cref{sec:rezk-compl-tensor}, we explain the strategy in detail for the subgoal of equipping $\DD$ (\resp $\HH : \CC \to \DD$) with a tensor (\resp tensor-preserving structure).

\subsection{The Rezk completion of a category with a tensor}
\label{sec:rezk-compl-tensor}

Let $\CC$ be a category and $\HH : \CC \to \DD$ a Rezk completion of $\CC$. 
Let $\tensor : \CC \times \CC \to \CC$ be a functor. 

In this section we equip $\DD$ with a functor $\tensorD : \DD \times \DD \to \DD$ such that

\begin{enumerate}
\item $\HH$ has the structure of a \emph{strong tensor-preserving} functor,
\ie, we have a natural isomorphism $\pt{\HH} : (\HH \times \HH) \comp \tensorD \Rightarrow \tensor \comp \HH$.
\item The precomposition functor of $(\HH,\pt{\HH})$ is an isomorphism of categories.
\end{enumerate}

\begin{definition}
[\coqident{CategoryTheory.Monoidal.RezkCompletion.LiftedTensor}{TransportedTensor},
\coqident{CategoryTheory.Monoidal.RezkCompletion.LiftedTensor}{TransportedTensorComm}]
\label{lifted-tensor-definition}
The \emph{lifted tensor} $\tensorD$ on $\DD$ 
is the (unique) functor $\tensorD : \DD\times\DD \to \DD$ such that there is a natural isomorphism as depicted in the following diagram:
\[
\begin{tikzcd}
& {\DD \times \DD} 
	\arrow[dd, Rightarrow, shorten <=10pt,shorten >=10pt, "\pt{\HH}"] 
	\arrow[rd, "\tensorD"] & \\ 
{\CC \times \CC} 
	\arrow[ru, "{\HH\times\HH}"] 
	\arrow[rd, swap, "\tensor"] 
& & \DD \\
& \CC \arrow[ru, swap, "\HH"] &
\end{tikzcd}
\]
\end{definition}

\begin{remark}
The functor $\tensorD$ is given by applying \cref{weq-induces-iso-lemma} to the weak equivalence $\HH \times \HH : \CC \times \CC \to \DD \times \DD$.
\end{remark}

\begin{remark} The natural isomorphism is labelled as $\pt{\HH}$ because this natural isomorphism is precisely the structure we need to have that $\HH$ is a (strong) tensor-preserving functor.
\end{remark}

\begin{lemma} 
[\coqident{CategoryTheory.Monoidal.RezkCompletion.LiftedTensor}{HT_eso}]
\label{precomp-tensor-eso-lemma}
Let $\EE$ be a univalent category and $\tensor_\EE : \EE\times\EE\to\EE$ be a functor.
The displayed precomposition functor (\cref{disp-precomp-functor-definition}) $\precomp{\pt{\HH}}$ with target displayed object $\tensor_\EE$ (as a displayed object in $\CATT$) is displayed split essentially surjective.
Consequently, the precomposition functor
\[
\precomp{(\HH,\pt{\HH})} : \mor{\textstyle\int\CATT\,}{(\DD,\tensorD)}{(\EE,\tensor_E)} \to \mor{\textstyle\int \CATT\,}{(\CC,\tensor)}{(\EE,\tensor_E)}
\]
is essentially surjective.
\end{lemma}
\begin{proof}
Let $G : \DD \to \EE$ be a functor and $\pt{\HH \comp G}$ a natural transformation of type 
\[
(\HH\times\HH) \comp (G \times G) \comp \tensor_\EE \Rightarrow \tensor \comp \HH \comp G\enspace.
\]
witnessing that $\HH \comp G$ is a lax tensor-preserving functor. 
We have to construct a natural transformation witnessing that $G$ is a lax tensor-preserving functor, \ie, we have to define a natural transformation
\[
\pt{G} : (G \times G) \comp \tensor_\EE \Rightarrow \tensorD \comp G\enspace.
\]
Since $\HH \times \HH$ is a weak equivalence and $\EE$ is univalent, 
it suffices to define a natural transformation of type
\[
(\HH\times\HH) \comp (G \times G) \comp \tensor_\EE \Rightarrow (\HH\times\HH) \comp \tensorD \comp G\enspace.
\]
which we define as:
\[
\begin{tikzcd}
& {\DD \times \DD} \arrow[r, "G\times G"]
& {\EE \times \EE}
	\arrow[dl, Rightarrow, shorten <=10pt,shorten >=10pt, swap, "\pt{\HH \comp G}"]
	\arrow[rd, "\tensor_\EE"] & \\
{\CC \times \CC}
	\arrow[ru, "{\HH\times\HH}"]
	\arrow[r, swap, "\tensor"]
	\arrow[rd, swap, "{\HH\times\HH}"]
& \CC
	\arrow[r, swap, "\HH"]
	\arrow[d, Rightarrow, shorten <=1pt,shorten >=1pt, "(\pt{\HH})^{-1}"]
& \DD \arrow[r, swap, "G"] & \EE \\
& {\DD \times \DD} \arrow[ru, swap, "\tensorD"] &
\end{tikzcd}
\]
For a detailed proof that $\pt{\HH \comp G}$ is (displayed) isomorphic to the (displayed) composition of $\pt{\HH}$ and $\pt{G}$, we refer the reader to the formalization.
\end{proof}

\begin{lemma} 
[\coqident{CategoryTheory.Monoidal.RezkCompletion.LiftedTensor}{HT_ff}]
\label{precomp-tensor-ff-lemma}
Let $\EE$ be a univalent category and $\tensor_\EE : \EE\times\EE\to\EE$ be a functor.
The displayed precomposition functor $\precomp{\pt{\HH}}$ is displayed fully faithful.
Consequently, the precomposition functor $\precomp{(\HH,\pt{\HH})}$ between the tensor-preserving functor categories is fully faithful.
\end{lemma}
\begin{proof}
It is displayed faithful because the type stating that a natural transformation preserves a tensor is a mere proposition.
In order to show that it is displayed full, notice that we have to show an equality of morphisms, \ie, a proposition. 
Therefore, we are able to use that $\HH\times\HH$ is essentially surjective which allows us to work with objects in $\CC$ instead of $\DD$ which leads to the result.
\end{proof}

\begin{theorem} 
[\coqident{Bicategories.MonoidalCategories.MonoidalRezkCompletion}{precomp_tensor_catiso}]
\label{rezk-completion-tensor-theorem}
A category equipped with a tensor admits a Rezk completion: 
Let $(\EE,\tensor_\EE) : \int \CATT$. If $\EE$ is univalent, then
\[
\precomp{(\HH,\pt{\HH})} : \mor{\textstyle\int\CATT\,}{(\DD,\tensorD)}{(\EE,\tensor_E)} \to \mor{\textstyle\int \CATT\,}{(\CC,\tensor)}{(\EE,\tensor_E)}
\]
is an isomorphism of categories.
\end{theorem}
\begin{proof}
First notice that both categories are univalent, indeed:
since $\EE$ is univalent, so are $\mor{\CAT}{\DD}{\EE}$ and $\mor{\CAT}{\CC}{\EE}$
and in \cref{sec:bicat-mono-categ}, we have proven that the displayed bicategory $\univCATT$ is locally univalent, \ie, the displayed hom-categories are univalent.
Hence, it suffices to show that this functor is a weak equivalence, \ie, fully faithful and essentially surjective.
Fully faithfulness can always be concluded if both the functor on the base categories and the displayed functor are.
The total functor is essentially surjective if this holds on the base and at the displayed level, provided extra information: it suffices that the base category and the displayed category are univalent.
So we conclude the result from combining the assumption that $\HH$ is a weak equivalence and lemmas \ref{precomp-tensor-ff-lemma} and \ref{precomp-tensor-eso-lemma}.
\end{proof}

\begin{remark}
The strategy introduced in this \lcnamecref{sec:rezk-compl-tensor} will be repeated in the next section, so we refer back to this section for the necessary details (if needed).
\end{remark}

\subsection{The Rezk completion of a category with a tensor and unit}
\label{sec:rezk-compl-tensor-unit}
In \cref{sec:rezk-compl-tensor}, we have shown how the structure of a tensor $\tensor$ on $\CC$ transports along a weak equivalence $\HH:\CC\to\DD$ to a tensor on a univalent category $\DD$.
Furthermore, $\HH$ has the structure of a strong tensor-preserving functor and that $(\DD,\tensorD)$ is universal in the sense that objects in $\int \CATT$ admit a Rezk completion.

In this \lcnamecref{sec:rezk-compl-tensor-unit}, we show that the same result holds when we add the choice of an object to a category, playing the role of the tensorial unit. 
This construction is trivial, but we will also discuss how we can conclude that objects in $\int \CATTU$ admit a Rezk completion.

As before, let $\HH : \CC\to \DD$ be a weak equivalence from a category $\CC$ to a univalent category $\DD$. Let $\unit : \CC$, thus $(\CC,\unit) : \int \CATUnit$. 
Clearly we have $(\HH,\id{\HH\,\unit}) : \mor{\int \CATUnit}{(\CC,\unit)}{(\DD, \HH\,\unit)}$.

To conclude that $(\DD,\HH\,\unit)$ is universal, we apply the same reasoning as in \cref{sec:rezk-compl-tensor}.
We have to show that for any $(\EE, \unit_\EE) : \CATUnit$ with $\EE$ univalent,
the displayed precomposition functor 
\[
\precomp{\id{\HH\,\unit}} : \mor{\CATUnit}{\HH\,\unit}{\unit_\EE} \to \mor{\CATUnit}{\unit}{\unit_\EE}
\]
is displayed fully faithful and displayed split essentially surjective.
We denote $\unitD := (\HH\,\unit)$ and $\pu{\HH} := \id{\unitD}$.

\begin{lemma} 
[\coqident{CategoryTheory.Monoidal.RezkCompletion.LiftedTensorUnit}{HU_eso}]
\label{precomp-unit-eso-lemma}
The displayed precomposition functor (\cref{disp-precomp-functor-definition}) $\precomp{\pu{\HH}}$ with target displayed object $\unit_\EE$ is displayed split essentially surjective.
Consequently, the precomposition functor $\precomp{(\HH, \pu{\HH})}$ with target object $(\EE, \unit_\EE)$ between unit tensor-preserving functor categories is essentially surjective.
\end{lemma}
\begin{proof}
It is merely surjective since the witness, expressing that the weak equivalence preserves the unit, is an identity morphism.
\end{proof}

\begin{lemma} 
[\coqident{CategoryTheory.Monoidal.RezkCompletion.LiftedTensorUnit}{HU_ff}]
\label{precomp-unit-ff-lemma}
The displayed precomposition functor  $\precomp{\pu{\HH}}$ is displayed fully faithful.
Consequently, the precomposition functor $\precomp{(\HH, \pu{\HH})}$ between the unit-preserving functor categories is fully faithful.
\end{lemma}
\begin{proof}
It is displayed faithful since the type of $2$-cells is a property.
The witness expressing that the weak equivalence preserves the unit is an identity morphism.
Hence, it is displayed full.
\end{proof}

Using the exact same reasoning used in \cref{rezk-completion-tensor-theorem}, we conclude:
\begin{theorem} 
[\coqident{Bicategories.MonoidalCategories.MonoidalRezkCompletion}{precomp_unit_catiso}]
A category equipped with a unit admits a Rezk completion: 
Let $(\EE,\unit_\EE) : \int \CATUnit$. If $\EE$ is univalent, then
\[
\precomp{(\HH,\pu{\HH})} : \mor{\textstyle\int \CATUnit\,}{(\DD,\unitD)}{(\EE,\unit_\EE)} \to \mor{\textstyle\int \CATUnit\,}{(\CC,\unit)}{(\EE,\unit_\EE)}
\]
is an isomorphism of categories.
\end{theorem}

So we have proven that objects in $\CATT$ and $\CATUnit$ admit a Rezk completion.
From these results, we conclude that objects in $\CATTU$ admit a Rezk completion: 

\begin{theorem}
[\coqident{Bicategories.MonoidalCategories.MonoidalRezkCompletion}{precomp_tensorunit_catiso}]
\label{rezk-completion-tensorunit-theorem}
Let $(\EE,\tensor_\EE, \unit_\EE) : \CATTU$. 
If $\EE$ is univalent, then
\[
\precomp{(\HH,\pt{\HH},\pu{\HH})} : \mor{\textstyle\int \CATTU\,}{(\DD,\tensorD,\unitD)}{(\EE,\tensor_\EE,\unit_\EE)} \to \mor{\textstyle\int \CATTU\,}{(\CC,\tensor,\unit)}{(\EE,\tensor_\EE, \unit_\EE)}
\]
is an isomorphism of categories, \ie, objects in $\int \CATTU$ admit a Rezk completion.
\end{theorem}
\begin{proof}
The product of univalent displayed bicategories is again univalent.
Thus, both the domain and codomain of this functor are univalent.
Hence, by the same argument as in \cref{rezk-completion-tensor-theorem}, it reduces to proving that the displayed precomposition functor is a displayed weak equivalence.
The displayed precomposition functor is the product of the displayed precomposition functors of $\pt{\HH}$ \resp $\pu{\HH}$.
Since the product of displayed weak equivalences is again a weak equivalence, the result now follows.
\end{proof}

\subsection{The Rezk completion of a category with a tensor, unit, unitors and associator}
\label{sec:rezk-compl-tensor-unit-unitor-associator}
In this \lcnamecref{sec:rezk-compl-tensor-unit-unitor-associator}, we prove that every object in $\int \CATLU$ (\resp $\int \CATRU$ and $\int \CATA$) has a Rezk completion.

As above, we let $\HH : \CC \to \DD$ be a weak equivalence from a category $\CC$ to a univalent category $\DD$, and let $\CC$ be equipped with a tensor $\tensor$ and a unit $\unit$.
The lifted tensor on $\DD$ is denoted by $\tensorD$ and $\unitD := \HH\, \unit$.
The witness that $\HH$ preserves the tensor (\resp unit) (strongly) is denoted by $\pt{\HH}$ (\resp $\pu{\HH} = \id{\HH\, \unit}$).

\begin{remark}
  In all the constructions of this section, we use the lifted tensor $\tensorD$ and unit $\unitD$.
  The specific shape of these lifts does not matter; we could state the constructions for an \textit{arbitrary} Rezk completion of $\CATTU$. However, by univalence we have uniqueness of the tensor and unit on $\DD$ under the proviso that $\HH$ preserves them both.
\end{remark}

Before lifting a left unitor from $\CC$ to $\DD$, we first define a natural isomorphism witnessing that the weak equivalence preserves tensoring with the unit object (on the left):
\begin{lemma}
[\coqident{CategoryTheory.Monoidal.RezkCompletion.LiftedUnitors}{LiftPreservesPretensor}]
\label{lift-preserves-pretensor-lemma}
There is a natural isomorphism $\HH \comp (\unitD \mathop{\tensorD} -) \Rightarrow (\unit \tensor -) \comp \HH$.
\end{lemma}
\begin{proof}
This is given by the following composition:
\[
\begin{tikzcd}
\CC 
	\arrow[r, "\HH"] 
	\arrow[d,swap, "{(\unit, -)}"] 
& \DD 
	\arrow[d, "{(\unitD,-)}"] \\
{\CC \times \CC} 
	\arrow[r, "{\HH \times \HH}", ""{name=HH, above}] 
	\arrow[d,swap, "\tensor"] 
& {\DD\times\DD} 
	\arrow[d, "\tensorD"] \\
	\CC \arrow[r,swap, "\HH", ""{name=H, below}] 
& \DD
\arrow [Rightarrow , 
	shorten <=5pt,shorten >=5pt,
	"{\pt{\HH}}" swap, 
	from = HH , to = H ]
\end{tikzcd}
\]
where the upper square is given by a trivial equality of functors.
\end{proof}

\begin{definition} 
[\coqident{CategoryTheory.Monoidal.RezkCompletion.LiftedUnitors}{TransportedLeftUnitor}]
Let $\lu$ be a left unitor on $(\CC, \tensor,\unit)$, \ie, $(\CC,\tensor,\unit,\lu) : \int \CATLU$.
The \emph{lifted left unitor} $\luD$ on $(\DD,\tensorD,\unitD)$ is the unique natural isomorphism that maps to the vertical composition of the natural isomorphism (defined in \cref{lift-preserves-pretensor-lemma}) and $\lu \rwhisker \HH$, under the precomposition functor with $\HH$.
\end{definition}

An immediate calculation shows:
\begin{lemma}
[\coqident{CategoryTheory.Monoidal.RezkCompletion.LiftedUnitors}{H_plu}]
\label{weakequiv-preserves-lunitor-lemma}
$\HH$ preserves the left unitor. 
\end{lemma}

\begin{theorem} 
[\coqident{Bicategories.MonoidalCategories.MonoidalRezkCompletion}{precomp_lunitor_catiso}]
\label{thm:catlu-rc}
The objects in $\int \CATLU$ admit a Rezk completion:

Let $(\EE,\tensor_\EE, \unit_\EE,\lu_\EE) : \int \CATLU$. 
If $\EE$ is univalent, then $\precomp{(\HH,\pt{\HH},\pu{\HH},\plu{\HH})}$ of type
\[
\mor{\textstyle\int \CATLU\,}{(\DD,\tensorD,\unitD,\luD)}{(\EE,\tensor_\EE,\unit_\EE,\lu_\EE)} \to \mor{\textstyle\int \CATLU\,}{(\CC,\tensor,\unit,\lu)}{(\EE,\tensor_\EE, \unit_\EE,\lu_\EE)}
\]
is an isomorphism of categories, where $\plu{\HH}$ is a witness that $\HH$ preserves the left unitor (as provided by \cref{weakequiv-preserves-lunitor-lemma}).
\end{theorem}
\begin{proof}
As before, it reduces to show that the displayed precomposition functor (\cref{disp-precomp-functor-definition}) is a displayed weak equivalence.
It is displayed fully faithful since the type of $2$-cells in $\CATLU$ is the unit type.
We now show that it is displayed split essentially surjective. 
Let $G:\DD\to \EE$ be a lax tensor and unit preserving functor such that $\HH \comp G$ preserves the left unitor. We have to show that $G$ also preserves the left unitor.
Since we have to show a proposition, the claim now follows from combining the essential surjectivity of $\HH$ and then applying the assumption on $\HH \comp G$.
\end{proof}

Completely analogous is the case of right unitor:
\begin{theorem} 
[\coqident{Bicategories.MonoidalCategories.MonoidalRezkCompletion}{precomp_runitor_catiso}]
\label{thm:catru-rc}
The objects in $\int \CATRU$ admit a Rezk completion.
\end{theorem}

In order to prove that every object in $\int \CATA$ has a Rezk completion, 
we use an analogous trick as is used for objects in, \eg, $\int \CATLU$.
An associator for $(\DD, \tensorD)$ is a natural isomorphism between functors of type $(\DD \times \DD)\times \DD \to \DD$. 
Since the product of weak equivalences is again a weak equivalence, such a natural isomorphism corresponds uniquely to a natural isomorphism between functors of type $(\CC \times \CC)\times \CC \to \DD$. 
Analogous to the constructions of the left and right unitor, the natural isomorphism (of type $(\CC \times \CC)\times \CC \to \DD$) is not given by $\ass \rwhisker \HH$ as this does not give us the correct type of functors.
In the case of the left unitor, we only had to provide a natural isomorphism to match the domain, but for the associator, we furthermore need a natural isomorphism to match the codomain.

\begin{theorem} 
[\coqident{Bicategories.MonoidalCategories.MonoidalRezkCompletion}{precomp_associator_catiso}]
\label{thm:cata-rc}
The objects in $\int \CATA$ admit a Rezk completion.
\end{theorem}

\subsection{The Rezk completion of a monoidal category}
In this section, we are able to conclude that the objects in $\MONCAT$ and $\MONCATS$ admit a Rezk completion.

In the previous sections, we have lifted all the structure of a monoidal category to a weakly equivalent univalent category.

However, it still remains to show that the lifted structure $(\DD,\tensorD,\unitD,\luD,\ruD,\assD)$ satisfies the properties of a monoidal category if $(\CC, \tensor,\unit,\lu,\ru,\ass)$ does.

\begin{lemma}
[\coqident{CategoryTheory.Monoidal.RezkCompletion.LiftedMonoidal}{TransportedTriangleEq},
\coqident{CategoryTheory.Monoidal.RezkCompletion.LiftedMonoidal}{TransportedPentagonEq}]
The lifted monoidal structure satisfies the pentagon and triangle equalities:
If the triangle (\resp pentagon) equality holds for $(\CC, \tensor,\unit,\lu,\ru,\ass)$,
then it also holds for $(\DD,\tensorD,\unitD,\luD,\ruD,\assD)$.
\end{lemma}

\begin{theorem}
[\coqident{Bicategories.MonoidalCategories.MonoidalRezkCompletion}{precomp_monoidal_catiso}]
\label{rezk-completion-monoidal-theorem}

Any monoidal category admits a Rezk completion (considered in the bicategory of lax monoidal functors).
\end{theorem}
\begin{proof}
In \cref{thm:catlu-rc}, \cref{thm:catru-rc} and \cref{thm:cata-rc} we have shown how the categories $\int \CATLU$, $\int \CATRU$ and $\int \CATA$ admit a Rezk completion. 
Hence, $\int (\CATLU \times \CATRU \times \CATA)$ admits a Rezk completion.

Thus, to conclude that the total bicategory of $\CATP$ (over $\int (\CATLU \times \CATRU \times \CATA)$) admits a Rezk completion, it suffices to show that the displayed precomposition functor with respect to $\CATP$ is displayed fully faithful and displayed split essentially surjective. 
The displayed hom-categories of $\CATP$ are the terminal categories.
Hence, the displayed precomposition functor must be the displayed identity functor.
Consequently, this displayed precomposition functor is a weak equivalence.

\end{proof}

\begin{remark}
[\coqident{CategoryTheory.Monoidal.RezkCompletion.MonoidalRezkCompletion}{RezkCompletion_monoidal_cat},\coqident{CategoryTheory.Monoidal.RezkCompletion.MonoidalRezkCompletion}{RezkCompletion_monoidal_functor}]
\label{dfn:weq_transport_mon}
As part of the proof of \cref{rezk-completion-monoidal-theorem}, we have shown how to transfer a monoidal structure along a weak equivalence of categories, provided that the target category is univalent.
More precisely, for any monoidal category $\CC$, univalent category $\DD$, and weak equivalence $\HH : \CC \to \DD$, we construct a  monoidal structure $M$ on $\DD$, and a structure of a (strong) monoidal functor on $\HH$ with respect to $\CC$ and $M$.
\end{remark}

Next, we prove that any monoidal category admits a Rezk completion in the bicategory of strong monoidal functors. Concretely, we show the following theorem:
\begin{theorem}
[\coqident{Bicategories.MonoidalCategories.MonoidalRezkCompletion}{precomp_strongmonoidal_catiso}]
Let $\CC$ be a monoidal category and $\HH : \CC \to \DD$ the Rezk completion of $\CC$ as constructed in \cref{rezk-completion-monoidal-theorem}.
If $\EE$ is a univalent monoidal category, then
\[
\precomp{\HH} : \mor{\MONCATS}{\DD}{\EE} \to \mor{\MONCATS}{\CC}{\EE}
\]
is an isomorphism of categories.

\end{theorem}
\begin{proof}
First note that $\HH$ is indeed strong monoidal by the definition of $\pt{\HH}$ and $\pu{\HH}$. Hence, the statement is well-defined.

As before, we have to conclude that the displayed precomposition functor (\cref{disp-precomp-functor-definition}) $\precomp{((\pt{\HH})^{-1},(\pu{\HH})^{-1})}$ is fully faithful and displayed split essentially surjective.

The displayed precomposition functor is fully faithful since every type of displayed $2$-cells in $\MONCATS$ is the unit type.

The displayed precomposition functor is split essentially surjective since the lift of a natural isomorphism is a natural isomorphism.
\end{proof}

\subsection{The Rezk completion of a monoidal category using Day convolution}
\label{sec:day-conv}
A concrete implementation of the Rezk completion of a category $\CC$ is given by restricting the Yoneda embedding to its full image~\cite[Thm.~8.5]{DBLP:journals/mscs/AhrensKS15}.
It is well-known that any monoidal structure on $\CC$ induces a monoidal structure on its category of presheaves $[\CC^{op}, \SET]$~\cite[Prop.~4.1]{IM198675}. 
The tensor product of two presheaves $F,G$ is given by the \textit{Day convolution} $F \tensor_\mathsf{Day} G$. 
Furthermore, the Day convolution of representable presheaves is again representable, \ie, for any two objects $x,y : \CC$, one can construct a natural isomorphism
\[
\mor{\CC}{-}{x} \tensor_\mathsf{Day} \mor{\CC}{-}{y} \cong \mor{\CC}{-}{x \tensor y}\enspace.
\]
Consequently, the Yoneda embedding has the structure of a strong monoidal functor. 
As one would expect, the full subcategory of representable presheaves becomes the monoidal Rezk completion. 
One way to show this result is to show that the universal property of monoidal Rezk completion holds. 
However, we already know that the full subcategory of representable presheaves has a monoidal structure (induced by the monoidal Rezk completion).
Therefore, it suffices to show that the \textit{Rezk monoidal structure} is equal to the \textit{Day monoidal structure}. 

Each piece of data of the \textit{Rezk monoidal structure} is defined using a \textit{universal property} in the sense that it is a unique lifting of some functor or natural transformation. 
For example, the (lifted) tensor product $\tensorD$ is the unique functor satisfying the equation
\[
\tensor \comp \yon 
= (\yon \times \yon) \comp \tensorD\enspace,
\]
where $\yon$ is the Yoneda embedding restricted to its full image, \ie, the concrete weak equivalence.
Using that a category of presheaves is univalent, the Day tensor product also satisfies this equation. Hence, the Day tensor product and the lifted tensor coincide. 
The lifted unit is by definition equal to the unit of the \textit{Day monoidal structure}.
Analogously, one can argue that the Day unitors and associator also satisfy the \textit{universal property} of the lifted unitors \resp associator. 

This shows that, for the concrete implementation of the Rezk completion using representable presheaves, the monoidal Rezk completion is given by the Day convolution.

\begin{remark}
This section has briefly explained what one needs to do in order to work with a specific implementation of the Rezk completion of a category.
Indeed, Let $(\CC, \tensor, \unit, \lu,\ru,\ass)$ be a monoidal category and a \emph{specific} univalent category $\DD$ which is weakly equivalent to $\CC$ as witnessed by $\HH : \CC \to \DD$.
Furthermore, assume we have a functor $\tensorD : \DD\times\DD \to \DD$ and natural isomorphisms $\luD,\ruD$ and $\assD$ which have the types of a left unitor, right unitor and the associator (\wrt $\tensorD$ as the tensor and $\HH\, \unit$ as the unit).

Then, in order to show that $(\DD,\tensorD,\HH\,\unit,\luD,\ruD,\assD)$ is the monoidal Rezk completion, it suffices to show that the pieces of data satisfy the property of the lifted tensor, lifted left unitor, lifted right unitor and the lifted associator.
In particular, one \emph{does not} have to show manually that $(\DD,\tensorD,\HH\,\unit,\luD,\ruD,\assD)$ is a monoidal category, $\HH$ becomes a (strong) monoidal functor and that it satisfies the universal property of the monoidal Rezk completion; this all follows from the argument above.
\end{remark}

\section{Conclusion}
\label{sec:conclusion}

We have studied (the bicategory of) monoidal categories in univalent foundations.
First, we showed that the bicategory of univalent monoidal categories is univalent.
Second, we constructed a Rezk completion for monoidal categories; specifically, we lifted the Rezk completion for categories to the monoidal structure.
Our technique also works for lax and oplax monoidal categories, with minimal modifications.
We have not presented this work here, but the \UniMath code is available online.\footnote{\url{https://github.com/Kfwullaert/UniMath/tree/LaxMonoidalRezkCompletion}}

The second result provides a blueprint for constructing completion operations for ``categories with structure''.
By ``structure'', we mean categorical structure such as functors and natural transformations.
Here, the main challenge is to define a suitable notion of signature that allows us to specify structure on a category.
Such a signature should translate into a suitable ``tower'' of displayed (bi)categories and come with the necessary boilerplate code for using it.
Work on this topic will be reported elsewhere.

\bibliography{literature}

\end{document}